\documentclass[12pt,a4paper]{amsart}
\usepackage{times}
\usepackage[pdftex]{graphicx}
\usepackage{amssymb,amsfonts,amsmath,amsthm}
\usepackage{float}
\usepackage{xcolor}
\usepackage{enumitem}

\newtheorem{theorem}{Theorem}
\newtheorem{proposition}{Proposition}

\newtheorem{ass}{Assumption}
\newtheorem{definition}{Definition}
\newtheorem{remark}{Remark}

\newcommand\EE {\mathbb E}
\newcommand\FF {\mathbb F}

\newcommand\RR {\mathbb R}
\newcommand\PP {\mathbb P}

\newcommand\bW {\mathbf W}

\def\bone{\mathbf{1}}

\def\tmu{\tilde{\mu}}

\def\bW{\mathbb{W}}
\def\mT{\mathcal{T}}
\def\mX{\mathcal{X}}

\def\qed{\hskip6pt\vrule height6pt width5pt depth1pt}

\def\qed{\hskip 6pt\vrule height6pt width5pt depth1pt}

\newcommand{\ed}{\end{document}}
\newcommand{\be}{\begin{equation}}
\newcommand{\ee}{\end{equation}}
\newcommand{\bq}{\begin{eqnarray}}
\newcommand{\eq}{\end{eqnarray}}

\vspace{4in}

\newcommand{\bx}{\bar{x}}

\begin{document}

\title[Optimal Contract for a Fund Manager]{Optimal Contract for a Fund Manager, with Capital Injections and Endogenous Trading Constraints}
\author{Sergey~Nadtochiy and Thaleia~Zariphopoulou}
\thanks{S. Nadtochiy is partially supported by the NSF grant DMS-1651294.} 
\date{February 22, 2018
%\\\vskip 6pt Working draft -- confidential
}
\maketitle

\begin{abstract}
In this paper, we construct a solution to the optimal contract problem for delegated portfolio management of the fist-best (risk-sharing) type. The novelty of our result is (i) in the robustness of the optimal contract with respect to perturbations of the wealth process (interpreted as capital injections), and (ii) in the more general form of principal's objective function, which is allowed to depend directly on the agent's strategy, as opposed to being a function of the generated wealth only. In particular, the latter feature allows us to incorporate endogenous trading constraints in the contract. We reduce the optimal contract problem to the following inverse problem: for a given portfolio (defined in a feedback form, as a random field), construct a stochastic utility whose optimal portfolio coincides with the given one. We characterize the solution to this problem through a Stochastic Partial Differential Equation (SPDE), prove its well-posedness, and compute the solution explicitly in the Black-Scholes model.
%Remarkably, the optimal contract computed in the Black-Scholes model satisfies the limited liability condition and has additional properties which show that it also solves the second-best (moral hazard) version of the problem, in which the principal cannot fully deduce the agent's strategy from her observations.
\end{abstract}

\section{Introduction}

Herein, we study a problem of delegated portfolio management, in which an investor hires a fund manager (referred to as the agent) for a specified period of time, to invest her capital dynamically in the available assets. At the end of the time period, the investor receives the wealth generated by the manager and, in return, pays the fees prescribed by the contract. These fees are allowed to depend on the wealth level and on other publicly observed market indicators (e.g. on the prices of available assets). As the investor can deduce the agent's strategy from the generated wealth process and the publicly observed market indicators, the associated optimal contract problems is of the so-called ``first-best" type, also known as the optimal risk-sharing problem.

The existing literature on the optimal contract design for the delegated portfolio management problem, of the first-best type, includes \cite{Starks}, \cite{Stoughton}, \cite{OuYang}, \cite{CadenillasCvitanic}, and the references therein.\footnote{Herein, we limit our literature review to the papers that are dealing with the delegated portfolio management problem specifically, leaving aside the discussion of general optimal contract theory, such as the seminal work \cite{HolmstromMilgrom}.} Single period models are analyzed in \cite{Starks} and \cite{Stoughton}, while \cite{OuYang} considers the Black-Scholes-Merton model, with the investor and the fund manager having either exponential or power utilities. A general market model and general utilities are considered in \cite{CadenillasCvitanic}, which, in particular, constructs an optimal contract explicitly when the market is complete.

The present work differs from the existing results in that, herein, (i) we require that the contract is robust with respect to the perturbations of wealth process, and (ii) we consider a more general optimality criterion for a contract than the classical expected utility of terminal wealth.
Our main motivation to consider the perturbations of wealth process is to include (un-unticipated) \emph{capital injections} made by the investor after the contract is initiated. Namely, we assume that the contract allows the investor (as, e.g., most fee structures of mutual funds do) to add an arbitrary amount of additional capital to her account with the manager, at any time when she wishes to do so, and with the fee structure for the manager remaining the same (i.e. the contract remains the same). Note that these times and amounts, and even their probabilistic structure, may not be initially known to either one of the two parties. However, the inflow of capital in the fund may change the investment strategy of the fund manager drastically (see, e.g. \cite{BasakPavlova}, and the references therein, for more on the effects of capital inflows and outflows on the behavior of a fund manager). Thus, when designing an optimal contract, one needs to take into account the agent's optimal strategy, induced by this contract, for any intermediate time and wealth level. Mathematically, this means that the agent's strategy should be viewed as a random field, defined for all possible initial investment times and wealth levels.
%Another novel feature of the present setting is the possibility of future capital injections by the investor, whose times and sizes are completely unknown to the principal and to the agent at the time when the contract is designed. Namely, we assume that the contract allows the investor to add an arbitrary amount of additional capital to her account with the manager, at any time when she has such an opportunity, and with the fee structure for the manager remaining the same (i.e. the contract remains the same). Note that these times and amounts, and even their probabilistic structure, may not be initially known to the principal and to the agent (and even to the investor, herself). Thus, both the agent and the principal optimize their objectives assuming the worst case scenario, which, for both of them, is the case when no opportunities for additional injections arise. However, the principal does realize that future capital injections are possible, and she aims to ensure that the endogenous constraints will be satisfied in any scenario. Notice that adding a capital to the account, in general, changes the optimal strategy of the agent (even though the contract remains the same). Hence, if the investor is not concerned with endogenous constraints, she may add capital to her account even if it forces the agent to violate these constraints. To prevent this from happening, the principal searches for a contract which optimizes her objective and ensures that the endogenous constraints will be satisfied after every possible capital injection. 

Another special feature of our setting is the more general optimality criterion for a contract. Namely, we assume that the entity designing the contract (referred to as the principal) may be concerned directly with the strategy used by the agent, in addition to the wealth generated by this strategy.\footnote{As explained in the next paragraph, the investor may not coincide with the principal, in our setting.} Our main motivating example of such preference structure is the case of \emph{constrained maximization} of expected utility of terminal wealth, with the constraint that no investment is made in certain assets. In such a case, the principal's objective contains an infinite penalty for investing in the prohibited assets, and the contract must be designed so that the agent follows this rule. For example, a regulator or the board of directors of a mutual fund may want to enforce a ban on investments in certain ``socially irresponsible" assets, or in the assets of companies subject to sanctions (we refer the reader to \cite{SociallyResponsible}, and the reference therein, for more on the so called ``socially responsible" funds). However, the principal cannot put such a rule into a contract directly, as she does not observe the agent's actions. Hence, these constraints need to be enforced implicitly, through the design of the contract, which can only depend on the generated wealth and on the publicly observed factors -- this is what we refer to as the \emph{endogenous constraints}.

Let us describe a specific setting in which the robustness of the contract with respect to capital injections and the endogenous constraints are important (a more detailed formulation is given in Section \ref{se:example}).
First, we assume that the principal, who designs the contract, may not coincide with the investor (at least, they may not coincide for the entire duration of the contract). For example, the fee structure of a mutual fund is very often prescribed a priori, and an individual investor can either take it or leave it.\footnote{The examples of cases where a contract is not fully designed by the party that initiates it are numerous. For example, a lease agreement for a residential property is often standardized, according to the local laws, and it may be rather costly for an individual to design a new contract. In addition, the laws may require that certain conditions are present in the contract or may prohibit certain conditions: e.g. giving the lessee the right to terminate the agreement.} In this case, the principal may be a regulator or the board of directors of the mutual fund.\footnote{Alternatively, the principal may be an initial investor, who enters into a long-term contract with the fund manager and passes on her wealth to the successors. The successors cannot withdraw funds before the deadline, but they may be allowed to add capital, keeping the fee structure as prescribed by the principal.} Even though the principal may not coincide with the investor, we assume that she aims to design the contract so that the investor is satisfied: e.g. the board of directors of a mutual fund wants to keep their investors happy, in order not to lose them to the competitors. At the same time, the principal also wants to ensure that the agent does not invest in the prohibited assets. Thus, the principal finds a strategy that maximizes the investor's expected utility of terminal wealth, subject to the constraint that no investment is made in the prohibited assets, and aims to design a contract (which is only allowed to depend on the generated wealth and on the publicly observable market factors) which would make this strategy optimal for the agent.
This task is complicated by the fact that the agent may perform capital injections, whose times and sizes are unknown (i.e. not modeled) initially. Namely, the investor, unlike the principal, may not be concerned about investing in prohibited assets, hence, she may perform a capital injection even if it encourages the agent to violate this constraint. Thus, the contract has to be chosen by the principal so that the agent has no incentive to violate the constraint even in the presence of capital injections -- this is what we refer to as the \emph{robustness with respect to capital injections}.

On the mathematical side, this paper solves the following inverse problem: given a regular enough random field, find a stochastic utility whose optimal investment strategy, in the feedback form, coincides with this random field. We characterize the solution through a linear stochastic partial differential equation (SPDE), prove its well-posedness, and compute the solution explicitly in the Black-Scholes model.

The rest of the paper is organized as follows. In Section \ref{se:problem.form}, we formulate the optimal contract problem precisely, in mathematical terms. Subsection \ref{subse:mrkt} is concerned with the market model, and Subsection \ref{subse:opt.cont} introduces the notions of admissible and optimal contracts. Section \ref{se:solution} presents a general solution to the problem, which reduces to the inverse problem of constructing an optimization criterion that generates a given optimal strategy (viewed as a random field), for all initial wealth levels. Proposition \ref{prop:FPP} connects this problem to a nonlinear SPDE, and Proposition \ref{prop:SPDE.main} shows how to linearize this SPDE and proves the well-posedness of the resulting equation. Finally, Theorem \ref{thm:main} connects these results to the optimal contract problem. In Section \ref{se:example}, we consider a specific setting in which the proposed notion of optimal contract is natural, and use the general results of preceding sections to construct an optimal contract in closed form, in the Black-Scholes model. Remarkably, the optimal contract constructed in Section \ref{se:example} depends only on the values of the wealth process and of the tradable assets at the terminal time. Hence, it also provides a solution to the second-best (moral hazard) version of the problem, in which the principal only observes the terminal values of the wealth and of the market and, hence, cannot fully deduce the agent's strategy from her observations.

\section{Problem formulation}
\label{se:problem.form}

\subsection{Market model and investment strategies}
\label{subse:mrkt}
We fix a stochastic basis $(\Omega,\FF,\PP)$ and assume that the publicly observed filtration $\mathbb{F}$ (also referred to as the market filtration) is an augmentation of the filtration generated by $W$, a standard Brownian motion in $\RR^d$.
In addition, we assume that the price process of traded assets $S=(S^1,\ldots,S^k)^T$ is an It\^o process in $\RR^k$ with positive entries, given by
\begin{equation}\label{eq.marketModel.1}
d\log S_{t}=\tilde{\mu}_t dt + \sigma^T_t dW_t,
\end{equation}
where the logarithm is taken entry-wise, $\tilde{\mu}$ is a locally integrable stochastic process with values in $\RR^k$, and $\sigma$ is a $d\times k$ matrix of locally square integrable processes, with $d\geq k$, and with linearly independent columns. The latter assumptions is interpreted as the absence of redundant assets. We use the notation "$A^T$" to denote the transpose of a matrix (vector) $A$.
For simplicity, we set the riskless interest rate to zero (equivalently, we work with discounted units).
We introduce the $d$-dimensional stochastic process $\lambda$, frequently called the market price of risk,
via
\begin{equation}\label{lamba.1}
\lambda_t := \left(\sigma^T_t\right)^+ \mu_t ,
\end{equation}
where $(\sigma^T_t)^+$ is the Moore-Penrose pseudo-inverse of the matrix $\sigma^T_t$, and $\mu$ is the drift of $S$: $\mu^i_t=\tmu^i_t + \|\sigma^i_t\|^2/2$, for $i=1,\ldots,k$, with $\sigma_t^i$ being the $i$-th column of $\sigma_t$.
In particular, we have $\sigma^T_t \lambda_t = \mu_t$.
The existence of such a process $\lambda$ follows from the assumption of absence of arbitrage in the model.
Denote by $\mX$ a set of pairs $(\xi,\tau)$, with $\tau\in\mT$ and $\xi\in L^0_+(\mathcal{F}_{\tau})$, where $\mT$ is the set of all $\FF$-stopping times, and $L^0_+(\mathcal{F}_{\tau})$ is the set of all positive $\mathcal{F}_\tau$-measurable random variables.
Starting from any $(\xi,\tau)\in\mX$, the cumulative wealth process $X^{\pi,\xi,\tau}$ is given by
\begin{equation}\label{eq.Xpi.def}
dX_{s}^{\pi,\xi,\tau}= \pi^T_s \sigma^T_s \lambda_s ds + \pi^T_s \sigma^T_s dW_{s},\quad s\in(\tau,T],\quad X^{\pi,\xi,\tau}_{\tau}=\xi,
\end{equation}
for any progressively measurable process $\pi$, representing the self-financing trading strategy, for which the above integrals are well defined.  We assume that $\pi$ is such that $X^{\pi,\xi,\tau}$ is a.s. strictly positive at all times. For each pair $(\xi,\tau)$, we fix a subset of such strategies $\mathcal{A}(\xi,\tau)$, and call any $\pi\in\mathcal{A}(\xi,\tau)$ $(\xi,\tau)$-admissible (or, just admissible, if the rest is clear from the context).
%We denote the set of attainable wealth processes $X^{\pi}$ by $\mathcal{X}$.

\begin{remark}\label{rem:rem.1}
It is possible to drop the restriction to strictly positive wealth processes. However, in this case, the assumptions on the optimal strategy $\pi^*$ and on the initial condition $U_0$, made in Subsection \ref{subse:solve.SPDE}, as well as the proof of Proposition \ref{prop:SPDE.main}, would change accordingly (cf. Remark \ref{rem:pos.wealth.rem2}).
\end{remark}

\subsection{Optimal contract}
\label{subse:opt.cont}

Consider an investor who hires an agent in order to invest her initial capital $X_0>0$ in the market described above.
The agent is offered a contract, which is represented by a measurable mapping $C:\Omega\times(0,\infty)\rightarrow\RR$, which maps the terminal value of a wealth process (produced by the agent, via a chosen trading strategy $\pi$) into the payment (received by the agent at time $T$). The agent is risk-neutral, in that he aims to maximize his expected objective:
\begin{equation}\label{eq.OptCont.agent}
\max \EE \,C(X^{\pi}_T),
\end{equation}
where the maximization is performed over all admissible strategies $\pi$, with $C$ being fixed. The agent will not enter into a contract if his expected payment does not reach a given level $u_0>0$.
%In addition, we impose a monotonicity constraint, which is explained in more detail at the end of this subsection.
We define an admissible contract as a contract for which the agent's optimization problem is well posed, and such that the participation constraint is satisfied.
\begin{definition}\label{def:AdmisCont}
We call $C$ an admissible contract if the following holds.
\begin{itemize}
\item For any $(\xi,\tau)\in\mX$ and any $\pi\in\mathcal{A}(\xi,\tau)$, $C(X^{\pi,\xi,\tau}_T)$ is absolutely integrable.
\item There exists a progressively measurable random field $\pi^*:[0,T]\times\Omega\times (0,\infty)\rightarrow\RR$, s.t.:
\begin{itemize}
\item for any $(\xi,\tau)\in\mX$, there exists a unique $X^{*,\xi,\tau}$ satisfying (\ref{eq.Xpi.def}), with $\pi=\pi^*(X^{*,\xi,\tau})$,
\item for any $(\xi,\tau)\in\mX$, $\pi^*(X^{*,\xi,\tau})\in\mathcal{A}(\xi,\tau)$,
\item $\EE\, C\left(X^{*,X_0,0}_T\right)\geq u_0$,
\item for any $(\xi,\tau)\in\mX$ and any $\pi\in\mathcal{A}(\xi,\tau)$,
$$
\EE \left(C(X^{\pi,\xi,\tau}_T)\mid\mathcal{F}_\tau \right)\leq \EE \left(C(X^{*,\xi,\tau}_T)\,\mid\,\mathcal{F}_\tau \right),\quad a.s.,
$$
and the equality is only possible if $\pi=\pi^*(X^{*,\xi,\tau})$ for a.e. $(t,\omega)$ in the stochastic interval $[\tau,T]$.
\end{itemize}
Any such strategy $\pi^*$ is called $C$-optimal.
\end{itemize}
\end{definition}

The special feature of the above definition, which differentiates it from the classical setup, is that the agent is allowed to re-evaluate his strategy at intermediate times, and starting from various wealth levels, which, in particular, may not coincide with the wealth generated by his strategy thus far. In addition, at each re-evaluation, the agent has to follow the exact strategy prescribed by the optimal random field: i.e. the optimal strategy is time-consistent and unique. A motivation for such strong definition of an optimal contract is given in the discussion following Definition \ref{def:OptCont}, and a specific problem is described in Section \ref{se:example}.

The contract is designed by a principal who aims to maximize the expectation of her individual objective $J$, which maps any progressively measurable random field $\pi:[0,T]\times\Omega\times (0,\infty)\rightarrow\RR$ into an $\mathcal{F}_T$-measurable random variable $J(\pi)$, applied to the strategy used by the agent, less the payment to the agent:
\begin{equation}\label{eq.OptCont.principal}
\max \EE \left[J(\pi) - C\left(X^{\pi}_T\right)\right].
\end{equation}
The above maximization is performed over all admissible contracts $C$, with the strategy $\pi$ being $C$-optimal.

\begin{definition}\label{def:OptCont}
An admissible contract $C^*$ is a solution to the optimal contract problem (\ref{eq.OptCont.agent})--(\ref{eq.OptCont.principal}), also referred to as an optimal contract, if, for any $C^*$-optimal strategy $\pi^*$, any admissible contract $C$, and any $C$-optimal $\pi$, we have
$$
\EE \left(J(\pi) - C\left(X^{\pi,X_0,0}_T\right)\right) \leq \EE \left(J(\pi^*) - C^*\left(X^{*,X_0,0}_T\right)\right),
$$
where $X^{\pi,X_0,0}$ and $X^{*,X_0,0}$ are the wealth processes associated with $\pi$ and $\pi^*$, respectively, and with the initial condition $X_0$ at time zero.
\end{definition}

%In the above definition of an optimal contract, we implicitly require that the principal must be happy with the agent's strategy regardless of which optimal strategy the agent chooses.
The main difference between the above formulation of the optimal contract problem and the classical one is that, in the present case, the principal needs to predict the agent's strategy for various initial wealth levels, which may not correspond to the levels generated by the strategy itself. The reason for such a formulation is explained in the introduction: on the one hand, we want to allow for (positive) capital injections after the contract is initiated, on the other hand, we do not want to impose any probabilistic structure on the times or the sizes of these injections. In such a robust formulation, the capital under management may change (increase) in an ``unpredictable way" at any given time, which, naturally, forces the agent to change his strategy. However, Definition \ref{def:AdmisCont} ensures that, even if an injection is made, the agent's optimal strategy is still given by the same random field (only started from a different wealth level). Thus, in the presence of unknown capital injections, the contract can only determine the agent's optimal strategy as a random field. This makes it natural to define the principal's objective as a function of such random field. A specific example that leads to an optimal contract problem of the present type is described in Section \ref{se:example}.

It is worth mentioning that, in the classical formulation of the problem, if we assume no capital injections and view strategies as stochastic processes, with a fixed initial wealth, the optimal contract problem typically reduces to the so-called ``first best" type, which has a trivial solution. This is due to the fact that, in a non-degenerate market, one can infer the trading strategy from a terminal value of the wealth process (viewed as a random variable). An example of such trivial construction is given in Subsection \ref{subse:fake.opt.cont}. However, the mapping from wealth to strategy (viewed as a stochastic process) depends on the initial capital, hence, the resulting, trivial, solution is not robust w.r.t. capital injections. The optimal contract defined above (with an example constructed in Subsection \ref{subse:ex.opt.cont}) is robust w.r.t. such injections, and it is also optimal in the classical formulation. Thus, in particular, it provides another, non-trivial, solution to the classical problem.
%This is illustrated by the example considered in Section \ref{se:example}.

It is also important that $J(\pi)$ may depend on $\pi$ in a more general way -- not only through $X^{\pi}$. Otherwise, the problem becomes trivial in many cases of interest, as illustrated in Subsection \ref{subse:fake.opt.cont}. As discussed in the introduction, our main motivation for considering general dependence on $\pi$ is the presence of endogenous constraints. Namely, we assume that the principal does not want the agent to invest in certain stocks but cannot simply include it in the contract, as the agent's strategy is not directly observable.

%Finally, let us discuss the monotonicity constraint, in the definition of admissible contracts. This constraint is motivated by the assumption that there is always a way to ``waste" money: that the agent can always decrease the wealth process by an arbitrary adapted process, which may, e.g., result from round-trip trades with commission, or other irrational investments. Assume, now, that there exist $X,X'\in C([0,T])$, s.t. $X_t\leq X'_t$, for all $t\in[0,T]$, and $C(X)>C(X')$. Then, by ``wasting" capital, if needed, the agent can modify his contract on $X'$ to take the same value as on $X$. Clearly, any such modification is beneficial for the agent, hence, he will not use the original contract when choosing his optimal strategy, which makes the optimal contract problem non-sensical. The assumption of monotonicity resolves this issue. This assumption can also be interpreted as a rule that a contract should not create incentives to waste capital. It is worth mentioning that, without the monotonicity constraint, if we assume that the agent cannot ``waste" capital, there often exists a trivial solution to the optimal contract problem, as shown in Section \ref{se:example}.

\begin{remark}
Note that we allow the principal's individual objective, $J$, and the contract, $C$, to be quite general. However, the principal's total objective combines them in the additive way: $J-C$. From an economic point of view, it may be more natural to include the agent's fees inside $J$, but it is not allowed by the current setting. Nevertheless, the subsequent sections show that the optimal contract is constructed as $C(x) = \overline{C}_T(x)$, where $\overline{C}$ is a sufficiently smooth random filed, so that we can define
$$
C(X^{\pi}_T) = \overline{C}_0(X_0) + \int_0^T d\, \overline{C}_t(X^{\pi}_t).
$$
As we assume no discounting (equivalently, we work with discounted units), the above representation can be interpreted as a flow of payments from the principal to the agent. As these payments are spread over the entire time interval $[0,T]$, it is possible to justify their appearance in the additive form in the principal's objective.
\end{remark}

\begin{remark}
The assumption of risk-neutrality of the agent can be relaxed by assuming that he maximizes the expected utility of his fees, $U(C)$. However, in such a case, we would either have to replace $C$ in the principal's objective by $U(C)$, or the agent's participation constraint would have to be formulated in terms of expected fees (as opposed to expected utility of his fees), none of which is very natural. In addition, we do not allow for a cost of effort in the agent's objective.
These are the limitations we have to accept in order to be able to use our solution approach. We leave the case of more general preferences and cost structures for future research.
\end{remark}

\begin{remark}
The optimal contract constructed herein is also robust w.r.t. maturity. Namely, our method allows one to construct an entire family of optimal contracts, $\{C_T\}$, for all maturities $T>0$. Thus, we also solve a slightly more general optimal contract problem (of the so-called ``third best" type), in which the agent is allowed to choose the time horizon (when the contract is initiated), and the principal does not know which horizon the agent prefers, hence, she offers him a menu of contracts, for all possible horizons.
\end{remark}

\begin{remark}
A very desirable feature of a contract is its limited liability: i.e. the condition $C\geq0$. Note that we do not require limited liability in the definition of admissible contract, and our general results do not guarantee that this property is satisfied by the optimal contract. However, the optimal contract constructed in Section \ref{se:example} does satisfy the limited liability condition.
\end{remark}

%\begin{remark}\label{rem:symmetry}
%Another desirable feature of a contract is its symmetry around the chosen benchmark. As required by the Amendment to the Investment Advisors Act of 1940, if the fee structure of a mutual fund (and many other financial organizations) is based on the performance of the fund relative to a chosen benchmark (typically, a given portfolio, but, more generally, a given stochastic process), then, the contract cannot be one-sided. Namely, if the contract prescribes a bonus for the fund manager in case of the outperformance relative to the benchmark, it must also prescribe a penalty for the underperformance. Even though we do not explicitly require such symmetry property in our definition of an optimal contract, the contract constructed in Section \ref{se:example} satisfies this property.
%\end{remark}

\section{Solution}
\label{se:solution}

Let us outline, heuristically, the solution approach.
First, we notice that, if $C$ is an admissible contract and $\pi^*$ is $C$-optimal, with the associated optimal wealth $X^*$, the contract
\begin{equation}\label{eq.C.norm}
\tilde{C}:=C\frac{u_0}{\EE\, C(X^{*}_T)}
\end{equation}
is also admissible, and the set of $\tilde{C}$-optimal strategies is the same as the set of $C$-optimal strategies. In addition,% for any optimal strategy $\pi^*$,
$$
\EE \,\tilde{C}(X^{*}_T)=u_0.
$$
Thus, there is no loss of optimality in restricting the candidate contracts $C$ to those admissible contracts for which $\EE C(X^{\pi})=u_0$, for every $C$-optimal $\pi$.
This implies that we can drop the expected payment to the agent in the principal's objective and solve the relaxed problem: find a random field $\pi^*$ and the associated optimal wealth $X^*$ (with initial condition $(X_0,0)$), s.t.
$$
\pi^*(X^*)\in\text{argmax}\, \EE \,J(\pi),
$$
where the maximization is performed over all $\pi\in\mathcal{A}(X_0,0)$. The main idea is to construct, for a given $\pi^*$, an admissible contract $C$, s.t. $\pi^*$ is the only $C$-optimal strategy.
Normalizing $C$ as in (\ref{eq.C.norm}), we obtain the desired optimal contract.

Thus, the construction of an optimal contract reduces to solving the following \emph{inverse problem}: given a strategy $\pi^*$ (viewed as a random field), find an admissible contract $C$, s.t., for any $(\xi,\tau)\in\mX$ and any $\pi\in\mathcal{A}(\xi,\tau)$,
$$
\EE \left(C(X^{\pi,\xi,\tau}_T)\mid\mathcal{F}_\tau \right)\leq \EE \left(C(X^{*,\xi,\tau}_T)\,\mid\,\mathcal{F}_\tau \right),\quad a.s.,
$$
and the equality is only possible if $\pi=\pi^*(X^{*,\xi,\tau})$ for a.e. $(t,\omega)$ in the stochastic interval $[\tau,T]$. Fortunately, a solution to such problem is offered by the so-called forward performance SPDE. In the remainder of this section, we describe this solution, given by a random field $(U_t(x))_{t\geq0,\,x>0}$, and show that
$$
C(x) = u_0\frac{U_T(x)}{U_0(X_0)},
$$
is the desired optimal contract.

\subsection{Forward performance SPDE}

Recall that the value function in the classical utility maximization problem, at least formally, solves the Hamilton-Jacobi-Bellman (HJB) equation. The following SPDE is an analog of the HJB equation in a non-Markovian case:
\begin{equation}\label{SPDE}
dU_t(x) = \frac{1}{2}\frac{\|\partial_xU_t(x)\lambda_t + (\sigma^T_t)^+ \sigma^T_t \partial_x a_t(x)\|^2}{\partial^2_{xx}U_t(x)}dt + a^T_t(x) dW_t,\quad t\in[0,T],\,x>0,
\end{equation}
where $a_t(x)$ is a $d$-dimensional vector of progressively measurable random functions, continuously differentiable in $x$, which is referred to as the volatility of the forward performance process.
An application of It\^o-Ventzel formula proves the following fact (cf. \cite{mz-spde}, \cite{zar-RICAM}, \cite{ElKaroui}, \cite{ElKaroui.2}). 

\begin{proposition}\label{prop:FPP}
Assume that $a=(a_t(x))_{t\in[0,T],\,x>0}$ and $U=(U_t(x))_{t\in[0,T],\,x>0}$, respectively, are once and twice continuously differentiable stochastic flows (in the sense of \cite{Kunita}), satisfying (\ref{SPDE}), and such that $U$ is strictly concave in $x$ (a.s. for all times). Then, the following holds. 
\begin{enumerate}
\item For any $(\xi,\tau)\in\mX$ and any $\pi\in\mathcal{A}(\xi,\tau)$, the process $\left(U_t\left( X_{t}^{\pi,\xi,\tau}\right)\right)_{t\in[\tau,T]}$ is a local supermartingale (in the sense that there exists a localizing sequence that makes it a supermartingale).\footnote{Throughout the paper, such process is always defined w.r.t. the filtration $(\mathcal{F}_{\tau\vee t})_{t\in[0,T]}$, and its value on $[0,\tau]$ is $U_{\tau}(\xi)$.}
\item Assume that there exists a progressively measurable random field $\pi^*$, satisfying a.s., for all $t\in[0,T]$,
\begin{equation} \label{portfolio-sde}
\sigma_t \pi _{t}^*(x) = -\frac{\lambda_t \partial_x U_{t}(x) + (\sigma^T_t)^+ \sigma^T_t \partial_x a_{t}(x) }{\partial^2_{xx} U_{t}(x) },
\,\,\,\,\,\,\,\,\,\,\,\,\forall x>0,
\end{equation}
and such that, for any initial condition $(\xi,\tau)\in\mX$, there exists a unique (strong) solution $X^{*,\xi,\tau}$ to
\begin{equation}\label{wealth-sde}
dX_{t}^{*,\xi,\tau} = \left(\sigma_t\pi^*_t(X^{*,\xi,\tau}_t)\right)^T \lambda_t dt + \left(\sigma_t \pi^*_t(X^{*,\xi,\tau}_t)\right)^T dW_{t},
\quad t\in[\tau,T],
\quad X^{*,\xi,\tau}_{\tau} = \xi.
\end{equation}
Then, $\left(U_t\left( X_{t}^{*,\xi,\tau}\right)\right)_{t\in[\tau,T]}$ is a local martingale.
\item Assume that the conditions of the previous two items are satisfied, and that, in addition, the aforementioned local martingale and local supermartingales are a true martingale and true supermartingales, respectively. Then, for any $(\xi,\tau)\in\mX$ and any $\pi\in\mathcal{A}(\xi,\tau)$,
$$
\EE \left(U_T(X^{*,\xi,\tau}_T)\mid \mathcal{F}_{\tau} \right) \geq \EE \left(U_T(X^{\pi,\xi,\tau}_T)\mid \mathcal{F}_{\tau} \right)\quad a.s.,
$$
and the equality is only possible if $\pi=\pi^*(X^{*,\xi,\tau})$ for a.e. $(t,\omega)$ in the stochastic interval $[\tau,T]$. 
\end{enumerate}
\end{proposition}
\begin{proof}
As mentioned above, the proof of the theorem follows easily from an application of It\^o-Ventzel formula to $U_t(X^{\pi,\xi,\tau}_t)$. A direct computation verifies the first two claims. For the last claim, we only need to notice that the drift of $U_t(X^{\pi,\xi,\tau}_t)$ is strictly negative unless $\pi_t=\pi^*(X^{\pi,\xi,\tau}_t)$, with $\pi^*$ given by (\ref{portfolio-sde}). Then, taking conditional expectations, we obtain the desired inequality.

\qed
\end{proof}

%Of course, according to the definition, the local supermartingale and martingale properties are not sufficient for $U$ to be a forward performance process.
%Therefore, having solved the above SPDE (\ref{SPDE}) and constructed the optimal wealth via (\ref{wealth-sde}), one still needs to verify that the resulting process is, indeed, a forward investment performance process (this is analogous to the verification procedure in the classical utility maximization theory).

The last item of the above theorem implies that $\pi^*(X^*)$ maximizes the criterion $\EE U_T(X^{\pi}_T)$ over all admissible strategies, provided it is, itself, admissible. Of course, to establish this, one needs to (i) solve the SPDE (\ref{SPDE}), (ii) ensure the existence of $\pi^*$ and $X^*$, and (iii) drop ``local" in the supermartignale and martingale properties.
One way to ensure that the local supermartingale $\left(U_t\left( X_{t}^{\pi}\right)\right)_{t\geq0}$ is a true supermartingale, is to construct $U$ so that $\inf_{t,x}U_t(x)$ is bounded from below by an absolutely integrable random variable, and to restrict the initial wealth to absolutely integrable random variables. Then, one can also show by a standard argument that the local martingale $\left(U_t\left( X_{t}^{*}\right)\right)_{t\geq0}$ is a true martingale if and only if its expectation at any time coincides with its initial value. Of course, there also exist other ways to address (iii).

To address (i) and (ii), one needs to solve (\ref{SPDE}). However, the latter equation presents numerous difficulties associated with its nonlinear nature and, even more importantly, with the fact that it has ``time running in a wrong direction" (cf. \cite{NadtochiyTehranchi}, for a more detailed discussion of the latter issue). To date, there exist no existence or uniqueness results for the solutions to (\ref{SPDE}) in its general form. Nevertheless, in the next subsection, we choose a specific form of the volatility process $a$ and show how to construct a unique solution to (\ref{SPDE}), for any given (sufficiently regular) strategy $\pi^*$, given as a random field. If, in addition, (iii) is resolved and $\pi^*(X^*)$ is admissible, we obtain a solution to the optimal contract problem formulated in Subsection \ref{subse:opt.cont}. Indeed, if $\pi^*$ is the optimal strategy of the principal (i.e. the strategy she would like the agent to follow), the associated $U_T(x)$, normalized appropriately, produces the desired optimal contract.

\subsection{Solving the forward performance SPDE}
\label{subse:solve.SPDE}

Assume that we are given a random field
\begin{equation*}
\pi^*: \left(\RR_+\times\Omega\times(0,\infty), \mathcal{P} \otimes \mathcal{B}\left((0,\infty)\right)\right)
\rightarrow \left(\RR, \mathcal{B}\left(\RR\right)\right),
\end{equation*}
where $\mathcal{P}$ is the sigma-algebra of progressively measurable sets.
As usual, we suppress the dependence upon $\omega\in\Omega$.
We assume that $\pi^*$ is a sufficiently smooth random field, with the precise assumptions stated below.
In this subsection, we construct a solution to (\ref{SPDE}), such that (\ref{portfolio-sde}) holds with the given $\pi^*$.
%\begin{equation}\label{SPDE.2}
%dU_t(x) = \frac{1}{2}\frac{\|\partial_xU_t(x)\lambda_t + (\sigma^T_t)^+ \sigma^T_t \partial_x a_t(x)\|^2}{\partial^2_{xx}U_t(x)} + a^T_t(x) dW_t,
%\end{equation}
%such that
%\begin{equation} \label{portfolio-sde.2}
%\sigma_t \pi _{t}^*(x) = -\frac{\lambda_t \partial_x U_{t}(x) + (\sigma^T_t)^+ \sigma^T_t \partial_x a_{t}(x) }{\partial^2_{xx} U_{t}(x) },
%\,\,\,\,\,\,\,\,\,\,\,\,\forall x>0
%\end{equation}

Assume that $U$ solves (\ref{SPDE}) and 
\begin{equation}\label{eq.a.def.new}
a_t(x)= a(t,x, U_t, \partial^2_{xx} U_t) := a_t(\bx) - \lambda_t \left(U_t(x)- U_t(\bx)\right) - \int_{\bx}^x \sigma_t \pi^*_t(y)\partial^2_{yy}U_t(y) dy,
\end{equation}
where $\bar{x}>0$ is a fixed constant, and $\left(a_t(\bx)\right)_{t\geq0}$ is an arbitrary locally square integrable process in $\RR^d$.
With such a choice, we have:
\begin{equation}\label{eq.a.def.new.1}
\partial_x a_t(x)=  -\sigma_t\pi^*_t(x)\partial^2_{xx}U_t(x) - \partial_xU_t(x)\lambda_t.
\end{equation}
Then, recalling that the columns of $\sigma_t$ are linearly independent, we obtain
\begin{equation}\label{eq.a.def.new.2}
\partial_xU_t(x)\lambda_t + (\sigma^T_t)^+ \sigma^T_t \partial_x a_t(x)=  -\sigma_t\pi^*_t(x)\partial^2_{xx}U_t(x),
\end{equation}
%$$
%\partial_x a_t(x) = \sigma_t b_t(x),
%$$
%\begin{equation}\label{eq.a1.def.3}
%\partial_xU_t(x)\lambda_t + \sigma_t(\sigma^T_t\sigma_t)^{-1} \sigma^T_t \sigma_t b_t(x)=  -\sigma_t \pi^*_t(x)\partial^2_{xx}U_t(x),
%\end{equation}
%\begin{equation}\label{eq.a1.def.4}
%\sigma_tb_t(x)=   -\sigma_t\pi^*_t(x)\partial^2_{xx}U_t(x) - \partial_xU_t(x)\lambda_t,
%\end{equation}
and (\ref{SPDE}) becomes
\begin{eqnarray}
&&\label{eq.linSPDE.new.1} d U_t(x) = \frac{1}{2} \|\sigma_t\pi^*_t(x)\|^2 \partial^2_{xx}U_t(x)  dt\\
&&+ \left(a_t(\bx) - \lambda_t \left(U_t(x)- U_t(\bx)\right) - \int_{\bx}^x \sigma_t \pi^*_t(y)\partial^2_{yy}U_t(y) dy\right)^TdW_t\nonumber
\end{eqnarray}
The following derivations (until Assumption \ref{ass:ass.1}) are heuristic and are meant to motivate the main result of this subsection, Proposition \ref{prop:SPDE.main}.
Introducing $V_t(x):=\partial_x U_t(x)$, we, formally, differentiate the above equation, to obtain
\begin{eqnarray}
&& \label{eq.V.SPDE.new} dV_t(x) = \frac{1}{2} \partial_x\left(\|\sigma_t\pi^*_t(x)\|^2 \partial_{x}V_t(x)\right) dt
- \left( \sigma_t \pi^*_t(x) \partial_x V_t(x) + \lambda_t V_t(x)\right)^T dW_t.
\end{eqnarray}
%where $B^1_t=-W^1_t$ and $B^2_t=W^2_t$.
Next, we introduce $R_t(x):=-\partial_x V_t(x)=-\partial^2_{xx} U_t(x)$, and, formally, differentiate the above equation, to obtain
\begin{eqnarray}
&&d R_t(x) = \frac{1}{2}\left[ \partial_x\left(\|\sigma_t\pi^*_t(x)\|^2 \partial_{x}R_t(x)\right)%\right.\\
%&&\left.
+ \partial_x\left(\|\sigma_t\pi^*_t(x)\|^2 \right) \partial_xR_t(x)\right.\nonumber\\
&& \label{eq.R.SPDE.new} 
\left.+ \partial^2_{xx}\left(\|\sigma_t\pi^*_t(x)\|^2 \right) R_t(x)\right] dt
- \left[ \sigma_t\pi^*_t(x) \partial_x R_t(x) + \left(\lambda_t + \sigma_t\partial_x\pi^*_t(x)\right) R_t(x)\right]^T dW_t,
\end{eqnarray}
with the deterministic initial condition $R_0(x) = -\partial^2_{xx} U_0 (x)$.
%Changing variables, from $x$ to $z=\log x$, we introduce
%\begin{equation*}
%\tR_t(z):=R_t(e^z) = -\partial^2_{xx} U_t(e^z),\quad \tpi_t(z):=e^{-z} \pi^*_t(e^z),\quad z\in\RR,
%\end{equation*}
%and notice that $\tR_t(z)$ is expected to satisfy
%\begin{eqnarray}
%&&\label{eq.SPDE.solvable} d \tR_t(z) = \frac{1}{2} \left[ \partial_z\left(\|\sigma_t\tpi_t(z)\|^2\partial_{z}\tR_t(z)\right)
%+ \left((\partial_z + 2)(\|\sigma_t\tpi_t(z)\|^2)\right) \partial_z \tR_t(z)\right.\\
%&&\phantom{???????????????????????}\left. + \left((\partial^2_z + 3\partial_z + 2)(\|\sigma_t\tpi_t(z)\|^2)\right) \tR_t(x) \right]dt\nonumber\\
%&&\phantom{????}- \left( \sigma_t\tpi_t(z) \partial_z \tR_t(z) + \left(\lambda_t + \sigma_t(\partial_z+1)\tpi_t(z) \right) \tR_t(z)\right) dW_t,\nonumber
%\end{eqnarray}
%with deterministic initial condition $\tV_0(z) = -\partial^2_{xx} U_0 (e^z)$.

\begin{ass}\label{ass:ass.1}
Assume that, almost surely, for each $t\geq0$, the function $\pi^*_t\left(\cdot\right)$ is five times continuously differentiable and
%the function
%\begin{equation*}
%\tilde{\pi}^*_t(z):=e^{-z}\pi^*_t(\exp(z))
%\end{equation*}
%satisfies, almost surely, for all $t\geq 0$,
\begin{equation*}
\sup_{z\in\RR} \left| \sum_{j=1}^k \sigma^{ij}_t(\partial_z)^m \left(e^{-z}\pi^{*j}_t\left(e^z\right)\right)\right| \leq \xi_t,\,\,\,\,\,\,\forall m=0,\ldots,5,\,\,i=1,\ldots,d,
\end{equation*}
for some progressively measurable stochastic process $\xi$ with locally bounded paths.
\end{ass}

For any function $\phi:\RR\rightarrow\RR$, $m$-times weakly differentiable, we define the norm
\begin{equation*}
\|\phi\|_m:=\left(\sum_{j=0}^m \int_{\RR} r^2(z) \left(\phi^{(j)}(z)\right)^2 dz \right)^{1/2},
\end{equation*}
with
\begin{equation}\label{eq.r.def}
r(z):= \exp\left(\eta \sqrt{1+z^2}\right), %\exp(-\eta\sqrt{1+x^2}),
\end{equation}
with some constant $\eta>1$.
%and some fixed constant $\eta>0$.
Following \cite{KrylovGyongy}, we define the weighted Sobolev space $\bW^m$ (consisting of $m$-times weakly differentiable functions from $\RR$ to $\RR$) as the closure of $C^{\infty}_0(\RR)$ in the $\|.\|_m$ norm.

\begin{ass}\label{ass:ass.2}
Assume that $U_0$ is strictly concave, $\partial^2_{xx} U_0(\exp(\cdot)), \log(-\partial^2_{xx} U_0(\exp(\cdot))) \in \bW^3$, and that $|\lambda|$ has locally integrable paths.
\end{ass}

We now present one of the main results of this paper.

\begin{proposition}\label{prop:SPDE.main}
Let $\pi^*$, $U_0$, $\sigma$, and $\lambda$, satisfy  Assumptions \ref{ass:ass.1} and \ref{ass:ass.2}. Then, there exists a unique random field $R$ which solves (\ref{eq.R.SPDE.new}), with the initial condition $R_0 = -\partial^2_{xx} U_0$, and is such that $R_t(\log\cdot)$ takes values in $\bW^3$. The random field $R_{\cdot}(\cdot)$ is almost surely continuous and strictly positive.

In addition, for any constant $\bar{x}>0$ and any locally square integrable $\RR^d$-valued process $\left(a_t(\bx)\right)_{t\geq0}$, the random field $\left(U_t(x)\right)_{t\geq0,\,x>0}$, given by
\begin{equation}\label{eq.U.def.prop}
U_t(x) = \zeta_t + \int_{\bar{x}}^x \int_y^{\infty} R_t(z) dz dy,
\end{equation}
with
\begin{eqnarray*}
&& d\zeta_t = -\frac{1}{2} \|\sigma_t\pi^*_t(\bx)\|^2 R_t(\bx) dt + a^T_t(\bx) dW_t,\quad \zeta_0=U_0(\bx),
\end{eqnarray*}
is strictly concave and strictly increasing in $x$, and satisfies (\ref{SPDE}), with the volatility $a$ given by (\ref{eq.a.def.new}).
Moreover, for the given $\pi^*$, (\ref{portfolio-sde}) holds, and there exists a unique solution to (\ref{wealth-sde}), for any $(\xi,\tau)\in\mX$.
\end{proposition}
\begin{proof}
First, we transform (\ref{eq.R.SPDE.new}) with the simple change of variables, $x=\exp(z)$, introducing $\tilde{R}_t(z):=R_t(e^z)$, and (\ref{eq.R.SPDE.new}) becomes
%\begin{eqnarray*}
%&&d R_t(e^z) = \frac{1}{2}\left[ e^{-z}\partial_z\left(\|\sigma_t\pi^*_t(e^z)\|^2 e^{-z}\partial_{z}R_t(e^z)\right)
%+ e^{-z}\partial_z\left(\|\sigma_t\pi^*_t(e^z)\|^2 \right) e^{-z}\partial_zR_t(e^z)\right.\nonumber\\ 
%&&\left. + e^{-z}\partial_{z}(e^{-z}\partial_z \left(\|\sigma_t\pi^*_t(e^z)\|^2 \right)) R_t(e^z)\right] dt\nonumber\\
%&& - \left[ \sigma_t\pi^*_t(e^z) e^{-z} \partial_z R_t(e^z) + \left(\lambda_t + \sigma_t e^{-z}\partial_z\pi^*_t(e^z)\right) R_t(e^z)\right]^T dW_t,
%\end{eqnarray*}
\begin{eqnarray}
&&d\tilde{R}_t(z)= \frac{1}{2}\left[ (\partial_z+1)\left(\|e^{-z}\sigma_t\pi^*_t(e^z)\|^2 \partial_{z}\tilde{R}_t(z)\right)
+ (\partial_z+1)\left(\|e^{-z}\sigma_t\pi^*_t(e^z)\|^2 \right) \partial_z\tilde{R}_t(z)\right.\nonumber\\ 
&&\label{eq.SPDE.solvable}\phantom{???????????????????????}
\left. + (\partial^2_{zz}+3\partial_z+2) \left(\|e^{-z}\sigma_t\pi^*_t(e^z)\|^2 \right) \tilde{R}_t(z)\right] dt\\
&& - \left[ e^{-z} \sigma_t\pi^*_t(e^z) \partial_z \tilde{R}_t(z) + \left(\lambda_t + (\partial_z+1)\left(e^{-z}\sigma_t\pi^*_t(e^z)\right)\right) \tilde{R}_t(z)\right]^T dW_t,\nonumber
\end{eqnarray}
%Let's show that the above SPDE has a unique generalized solution in the sense of Krylov-Rozovskii (1979).
Notice that the SPDE (\ref{eq.SPDE.solvable}) is linear and (degenerate) parabolic. In particular, it belongs to the class of equations analyzed in \cite{KrylovGyongy}.
Following this reference, we refer to Example 2.2 in \cite{KrylovGyongy}, and the preceding discussion, to conclude that the conditions of Theorem 2.5 in \cite{KrylovGyongy} are satisfied, with $m=3$ and $\Gamma=1$. The latter theorem states that there exists a unique generalized solution $\tilde{R}$ to (\ref{eq.SPDE.solvable}), with $\tilde{R}_0(z) = -\partial^2_{xx} U_0(e^z)$, which is a progressively measurable process with values in $\bW^3$, having continuous paths in $\bW^2$.
Notice that $\tilde{R}_t\in \bW^3$ implies that $\tilde{R}_t(.)$ is twice continuously differentiable. Hence, the random field $\tilde{R}_\cdot(\cdot)$ is almost surely continuous and strictly positive, and the spatial derivatives in (\ref{eq.SPDE.solvable}) can be understood in the classical sense. Then, changing the variables back to $x=\exp(z)$, we conclude that $R_t(x):=\tilde{R}_t(\log x)$ solves (\ref{eq.R.SPDE.new}). Reverting these arguments, we obtain uniqueness of the solution to (\ref{eq.R.SPDE.new}).

Next, we show that $R$ is strictly positive. Notice that it suffices to find a progressivley measurable random field $Y$, such that $\exp(Y)$ is a generalized solution to (\ref{eq.SPDE.solvable}), with the initial condition $-\partial^2_{xx} U_0(e^z)$. Then, from uniqueness, we conclude that $\tilde{R}$, and hence $R$, are positive.
To this end, we define $Y$ as the unique generalized solution to the following SPDE
%\begin{eqnarray}
%&&d Y_t(x) = \frac{1}{2} \left[  \partial_x\left(\tf^2_t(x)\partial_{x}Y_t(x)\right) - 2\lambda_t\tf_t(x) \partial_xY_t(x)
% + (\partial^2_x + 3\partial_x + 2)\tf^2_t\right.\nonumber\\
%&&\phantom{?????????????????????????????}\left. - (\lambda_t + (\partial_x+1)\tf_t(x))^2 - \tg_t(x)\right]dt\nonumber\\
%&&\label{eq.log.SPDE}\phantom{???????}+ \left( \tf_t(x) \partial_x Y_t(x) + \lambda_t + (\partial_x + 1)\tf_t(x)\right) dB^1_t + \tg_t(x) dB^2_t,
%\end{eqnarray}
\begin{eqnarray}
&&d Y_t(z) = \frac{1}{2} \left[ (\partial_z+1)\left(\|e^{-z}\sigma_t\pi^*_t(e^z)\|^2 \partial_z Y_t(z)\right) 
- 2\lambda^T_t e^{-z} \sigma_t \pi^*_t(e^z) \partial_z Y_t(z)\right.\nonumber\\
&&\label{eq.log.SPDE}\phantom{?????????????}
\left.+ (\partial^2_{zz}+3\partial_{z}+2)\|e^{-z}\sigma_t\pi^*_t(e^z)\|^2 - \|\lambda_t + (\partial_z+1)\left(e^{-z} \sigma_t \pi^*_t(e^z)\right)\|^2\right] dt\\
&&- \left[ e^{-z} \sigma_t\pi^*_t(e^z) \partial_z Y_t(z) + \lambda_t + (\partial_z+1)\left(e^{-z}\sigma_t \pi^*_t(e^z)\right) \right]^T dW_t,\nonumber
\end{eqnarray}
with the initial condition
\begin{equation*}
Y_0(z) = \log \tilde{R}_0(z) = \log(-\partial^2_{xx} U_0(e^z)).
\end{equation*}
%\begin{eqnarray*}
%&&d R_t(x) = \frac{1}{2}\left[ \partial_x\left(\|\sigma_t\pi^*_t(x)\|^2 \partial_{x}R_t(x)\right)%\right.\\
%+ \partial_x\left(\|\sigma_t\pi^*_t(x)\|^2 \right) \partial_xR_t(x) + \partial^2_{xx}\left(\|\sigma_t\pi^*_t(x)\|^2 \right) R_t(x)\right] dt\nonumber\\
%&& - \left[ \sigma_t\pi^*_t(x) \partial_x R_t(x) + \left(\lambda_t + \sigma_t\partial_x\pi^*_t(x)\right) R_t(x)\right] dW_t,
%\end{eqnarray*}
%assuming, of course, that $\tV_0(z)>0$.
%Making the additional assumption
%\begin{equation*}
%\|\log \tV_0\|^2_3<\infty,
%\end{equation*}
Theorem 2.5 in \cite{KrylovGyongy} states that the above equation has a unique generalized solution. Applying It\^o's formula, we deduce that $\exp(Y)$ solves (\ref{eq.SPDE.solvable}), with the initial condition $-\partial^2_{xx} U_0(e^z)$. The uniqueness of the solution implies $\exp(Y) = \tilde{R}$, hence, we conclude that $R$ is strictly positive.

Finally, we need to verify that the random field $U$, defined by (\ref{eq.U.def.prop}), is well defined and has the desired properties.
To this end, we define
\begin{equation*}
V_t(x)= \int_x^{\infty} R_t(y) dy.
\end{equation*}
Note that the above integral is well defined due to the choice of $r$ (cf. (\ref{eq.r.def})) and the fact that $\tilde{R}_t=R_t(\exp(\cdot))$ takes values in $\bW^3\subset\bW^0$: 
$$
\int_x^{\infty} R_t(y) dy
= \int_{\log x}^{\infty} e^z\tilde{R}_t(z) dz
\leq \left(\int_{\log x}^{\infty} r^2(z) \tilde{R}^2_t(z) dz \right)^{1/2} \left(\int_{\log x} e^{2z-2\eta \sqrt{1+z^2}} dz\right)^{1/2}<\infty
$$
Similarly, it is easy to deduce that $\partial_x R_t(\cdot)$ and $\partial^2_{xx} R_t(\cdot)$ are absolutely integrable over $(\varepsilon,\infty)$, for any $\varepsilon>0$.
Applying the stochastic Fubini theorem (cf. Theorem 64 in \cite{Protter.book}), we integrate (\ref{eq.R.SPDE.new}) to deduce that $V$ satisfies (\ref{eq.V.SPDE.new}), with the initial condition $V_0(x) = \partial_x U_0(x)$.\footnote{Strictly speaking, in order to apply Theorem 64 in \cite{Protter.book}, we need to localize $R$ and pass to the limit in the integrals over finite domain. We skip these routine arguments for the sake of brevity.}
Applying stchastic Fubini therem again, we integrate (\ref{eq.V.SPDE.new}), to show that $U$, defined by (\ref{eq.U.def.prop}), satisfies the SPDE (\ref{eq.linSPDE.new.1}).
It is clear that $U_t(\cdot)$ is strictly concave, as $R$ is strictly positive. Then, choosing $a_t$ via (\ref{eq.a.def.new}), we conclude that $U$ satisfies (\ref{SPDE}). In turn, equation (\ref{eq.a.def.new.2}) yields (\ref{portfolio-sde}).
Finally, Assumption \ref{ass:ass.1} implies that $\sigma_t\pi^*_t(\cdot)$ is globally Lipschitz, uniformly over $(t,\omega)$, which yields the existence and uniqueness of the solution to (\ref{wealth-sde}), for any initial condition $(\xi,\tau)\in\mX$.
\qed
\end{proof}

\begin{remark}
Proposition \ref{prop:SPDE.main} can be extended to hold with any positive weight function $r$, satisfying the condition ($\tilde{W}$) in \cite{KrylovGyongy}, and such that 
$$
\int_x^{\infty} \frac{e^{2z}}{r^2(z)} dz <\infty,\quad \forall\,x\in\RR.
$$
\end{remark}

\begin{remark}\label{rem:pos.wealth.rem2}
It is straight-forward to formulate the version of Proposition \ref{prop:SPDE.main} for the case where the wealth variable $x$ takes values in $\RR$ (as opposed to being restricted to $(0,\infty)$). This would correspond to the investment problems in which the wealth is not restricted to remain positive (cf. Remark \ref{rem:rem.1}). We did not find a unifying formulation that would allow us to treat both cases (i.e. $x\in\RR$ and $x>0$) simultaneously, and we chose to consider the case $x>0$. This choice is motivated by the example in Section \ref{se:example} which shows that, in the case $x>0$, in the Black-Scholes model, one can construct explicitly an optimal contract which also satisfies the limited liability condition. Currently, we do not know how to ensure the limited liability condition for the case $x\in\RR$, even in the context of this simple example.
\end{remark}

\begin{remark}
An alternative description of the solutions to (\ref{SPDE}), using duality methods, is given in \cite{ElKaroui}, \cite{ElKaroui.2}. However, the present construction is much shorter and more direct, and it allows us to obtain explicit solutions, as illustrated in Section \ref{se:example}. It is also worth mentioning that the Markovian solutions to (\ref{SPDE}) are analyzed in \cite{NadtochiyTehranchi}. 
\end{remark}

Propositions \ref{prop:FPP} and \ref{prop:SPDE.main} allow us to establish the following characterization of an optimal contract, which is the main result of this paper.

\begin{theorem}\label{thm:main}
Consider any initial capital $X_0>0$, as well as any $\lambda$ and $U_0$, satisfying Assumption \ref{ass:ass.2} and such that $U_0(X_0)>0$. 
Assume that there exists a progressively measurable random field $\pi^*$, such that $\pi^*$ and $\sigma$ satisfy Assumption \ref{ass:ass.1}, $\pi^*(X^{*,X_0,0})\in\mathcal{A}(X_0,0)$, and
$$
\EE J(\pi)\leq \EE J\left(\pi^*\right),
$$
for any $\pi$ that is $C$-optimal for some admissible contract $C$.
Let $U$ be defined as in Proposition \ref{prop:SPDE.main}, with any constant $\bar{x}>0$ and any locally square integrable $\RR^d$-valued process $\left(a_t(\bx)\right)_{t\geq0}$.
Then, the following holds.
\begin{enumerate}
\item For any $(\xi,\tau)\in\mX$ and any $\pi\in\mathcal{A}(\xi,\tau)$, the process $\left(U_t\left( X_{t}^{\pi,\xi,\tau}\right)\right)_{t\in[\tau,T]}$ is a local supermartingale.
\item For any $(\xi,\tau)\in\mX$, there exists a unique solution $X^{*,\xi,\tau}$ to (\ref{wealth-sde}), and the process $\left(U_t\left( X_{t}^{*,\xi,\tau}\right)\right)_{t\in[\tau,T]}$ is a local martingale. 
\item If the aforementioned local martingale and local supermartingales are a true martingale and true supermartingales, respectively, then,
$$
C^*(x):= U_T(x)\frac{u_0}{U_0(X_0)}
$$
is an optimal contract.
\end{enumerate}
\end{theorem}
\begin{proof}
Proposition \ref{prop:SPDE.main} implies that $U$, $a$, and $\pi^*$, satisfy all the assumptions of Proposition \ref{prop:FPP}. The first two statements of the theorem follow immediately.
To show the last statement, we notice that the admissibility of $\pi^*(X^*)$, the integrability of $C^*(X^{\pi})$, and the last part of Proposition \ref{prop:FPP}, imply that $C^*$ is an admissible contract and that $\pi^*$ is $C^*$-optimal.
%It is also clear that part 1 of Definition \ref{def:OptCont} is satisfied, with $\pi^*(X^*)$ in lieu of $\pi^*$.
To conclude, consider any admissible contract $C$ and any $C$-optimal $\pi$.
%$$
%\EE C'(X^{\pi})\leq \EE C'(X^{\pi'}),\quad \forall\,\pi\in\mathcal{A}.
%$$
Then, we have
$$
\EE \left[J(\pi) - C\left(X^{\pi,X_0,0}_T\right)\right] 
\leq \EE J(\pi) - u_0
\leq \EE J\left(\pi^*\right) - u_0
= \EE \left[J\left(\pi^*\right) - C^*\left(X^{*,X_0,0}\right)\right],
$$
where the first inequality follows from the admissibility of $C$ and the $C$-optimality of $\pi$, and the second inequality follows from the assumptions of the theorem.

\qed
\end{proof}

The next section illustrates the application of the above theorem. It describes a specific market model and a concrete contract design problem, for which the present definition of optimal contract is natural, and it shows how to construct an optimal contract explicitly. Moreover, the resulting optimal contract satisfies the limited liability condition: $C\geq0$ (note, however, that this condition is not guaranteed by Theorem \ref{thm:main}).

\section{Explicit optimal contract in the Black-Scholes model}
\label{se:example}

In this section, we assume that $d=k=2$, and
\begin{eqnarray*}
&&d\log(S^1_t) = (\mu_1 - \sigma^2_1/2)dt + \sigma_1 dW^1_t,\\
&&d\log(S^2_t) = (\mu_2 - \sigma^2_2/2)dt + \sigma_2 (\rho dW^1_t + \sqrt{1-\rho^2} dW^2_t),
\end{eqnarray*}
with some $\mu_1,\mu_2\in\RR$, $\sigma_1,\sigma_2>0$, and $\rho\in(-1,1)$.
In other words,
\begin{equation*}
\sigma=
\left(
\begin{array}{cc}
{\sigma_1} & {\sigma_2 \rho}\\
{0} & {\sigma_2\sqrt{1-\rho^2}}
\end{array}
\right),
\quad \mu=
\left(
\begin{array}{c}
{\mu_1} \\
{\mu_2}
\end{array}
\right),
\quad \lambda = (\sigma^T)^{-1} \mu=
\left(
\begin{array}{c}
{\mu_1/\sigma_1} \\
{\frac{\mu_2 - (\sigma_2\rho \mu_1)/\sigma_1}{\sigma_2\sqrt{1-\rho^2}}}
\end{array}
\right).
\end{equation*}
%We assume that the principal is only allowed to invest in $S^1$, while the agent may invest in both $S^1$ and $S^2$ (and both can, of course, invest in the riskless asset which yields zero interest, as all prices are measured in discounted units).

Let us fix a constant $\gamma\in(-\infty,0)\cup(0,1)$, whose meaning is explained below.
We let $\mX$ consist of all pairs $(\xi,\tau)$, s.t. $\tau$ is any stopping time with values in $[0,T]$ and $\xi,\xi^{\gamma}\in L^1\cap L^0_+(\mathcal{F}_\tau)$. For any $(\xi,\tau)\in\mX$, we define $\mathcal{A}(\xi,\tau)$ as the set of all locally integrable processes $\pi$, s.t. the resulting $X^{\pi,\xi,\tau}$ is strictly positive and
$$
\EE \sup_{t\in[\tau,T]} X^{\pi,\xi,\tau}_t\,+\,\EE \sup_{t\in[\tau,T]} \left(X^{\pi,\xi,\tau}_t\right)^{\gamma}<\infty.
$$

Next, consider an investor who is looking to hire an agent to manage her initial capital $X_0$.
As discussed in the introduction, we assume that the contract between the agent and the investor is designed by a third party, referred to as the principal (e.g., it can be a regulator, the board of directors of a mutual fund, etc.).
The principal chooses an optimal contract using the following individual objective:
\begin{equation}\label{eq.Princ.Obj.def}
J(\pi) = \frac{1}{\gamma} \left(X^{\pi,X_0,0}_T\right)^{\gamma}\,\bone_{\{\pi^2\equiv0\}} - \infty\cdot(1-\bone_{\{\pi^2\equiv0\}}),
\end{equation}
where $\pi$ is a random field, and $X^{\pi,X_0,0}$ is generated by (\ref{eq.Xpi.def}), with $\pi=\pi\left(X^{\pi,X_0,0}\right)$. The rationale behind this choice is as follows. The principal assumes (e.g., based on her estimates) that a typical investor uses power utility, with the relative risk aversion $1-\gamma$, and she adds the constraint that no investment can be made in $S^2$, as the latter asset is deemed inappropriate (e.g., immoral, subject to sanctions, etc.).
%The expected value of the above individual objective is to be maximized over all random fields $\pi$, s.t. the solution ti  $\pi\in\mathcal{A}(X_0,0)$.

Note that the investor may not be interested in the constraint $\pi^2\equiv0$ being met: e.g., in accordance with the assumption of the principal, she may aim to optimize the expected power utility, without the constraint. After the contract is initiated, the investor may have an opportunity to increase the size of her investment, at some stopping time $\tau$, to a random level $\xi$. As the investor may not care about the constraint $\pi^2\equiv0$, a priori, her capital injection may encourage the agent to violate this constraint. Neither the principal nor the agent are aware of the probabilistic properties of $(\xi,\tau)$ (i.e., we take the approach of Knightian uncertainty with regards to the opportunities of capital injections). In particular, after any capital injection $(\xi,\tau)$, the agent maximizes the expected value of the worst-case future scenario, which corresponds to no future opportunities for capital injections (as she can always choose not to use such an opportunity). Thus, after every capital injection $(\xi,\tau)\in\mX$, the agent solves
$$
\max_{\pi\in\mathcal{A}(\xi,\tau)} \EE \left(C(X^{\pi,\xi,\tau}_T)\mid \mathcal{F}_{\tau} \right).
$$

The regulator's task is two-fold. First, she needs to ensure that the investor is as happy with the contract as possible, given the constraint $\pi^2\equiv0$. Namely, the contract should be such that every optimal strategy of the agent maximizes the expectation of (\ref{eq.Princ.Obj.def}) less the expected payment to the agent, even in the presence of capital injections by the investor. Since these injections are not known to the regulator, she aims to maximize the worst case scenario for the investor, which is the case of no future opportunities for capital injections (as the investor can always choose not to use such an opportunity). This leads to the following objective for the regulator: find admissible contract $C^*$, s.t., for any $C^*$-optimal $\pi^*$, $(C^*,\pi^*)$ maximizes
\begin{equation}\label{eq.Regulator.obj}
\EE\left[ J(\pi) - C(X^{\pi,X_0,0}_T) \right]
\end{equation}
among all pairs $(C,\pi)$ with admissible $C$ and $C$-optimal $\pi$. 
It is easy to see that, if $C^*$ is an optimal contract, in the sense of Definition \ref{def:OptCont}, then it solves the first task of the regulator.
The second task of the regulator is to ensure that the investor will not encourage the agent to invest in the second asset by her capital injections.
%Note that an honest investor would never execute a capital injection that would encourage the agent to violate the constraint $\pi^2\equiv0$, as she cares about this constraint. A dishonest investor, a priori, may perform such action. 
This task is resolved by the admissibility property of an optimal contract $C^*$: cf. Definitions \ref{def:AdmisCont} and \ref{def:OptCont}. %Indeed, the admissibility guarantees that no capital injection will encourage the agent to violate the constraint.
Indeed, the admissibility implies that, after each capital injection, it is still optimal for the agent to follow the optimal strategy (understood as a random field) computed under the assumption of no capital injections. The latter strategy does not invest in $S^2$, as the pair $(C^*,\pi^*)$ maximizes the objective (\ref{eq.Regulator.obj}).
In the following subsections, we construct an optimal contract $C^*$ explicitly.

\subsection{Principal's optimal strategy}

Following the solution approach outlined at the beginning of Section \ref{se:solution}, we, first, search for a random field $\pi^{*1}$, s.t.
$$
\pi^{*1}(X^*)\in\text{argmax} \frac{1}{\gamma} \EE \left(X^{\pi,X_0,0}_T\right)^{\gamma},
$$
where $X^*$ is the associated optimal wealth (starting from $X_0$ at time zero), and the supremum is taken over all processes $\pi^{1}$, s.t. $\pi=(\pi^{1},0)^T\in\mathcal{A}(X_0,0)$.
The wealth process, in this case, satisfies
$$
X^{\pi,X_0,0}_{0}=X_0\in\RR,\quad
dX_{s}^{\pi,X_0,0}= \pi^1_s \sigma_1 \lambda_1 ds + \pi^1_s \sigma_1 dW^1_{s},\quad s\in[0,T].
$$
The solution to the above optimal investment problem is well known, but we briefly outline it here, for the sake of completeness.
The associated HJB equation for the value function $V$ is
$$
\partial_t V + \max_{\pi^1}( \pi^1_s \sigma_1 \lambda_1 \partial_x V + \frac{1}{2}(\pi^1)^2 \sigma^2_1\partial^2_{xx} V)=0,
\quad x>0,\,\,t\in(0,T),
\quad V(T,x)=x^{\gamma}/\gamma.
$$
This yields
\begin{equation}
V(t,x)=\frac{x^{\gamma}}{\gamma} \exp\left( (T-t) \frac{\lambda^2_1 \gamma}{2(1-\gamma)}\right),
\quad \pi^{*1}_t(x) = \frac{\lambda_1}{\sigma_1(1-\gamma)} x,
\label{HJB-1}
\end{equation}
\begin{equation}\label{eq.princ.opt.wealth}
X^{*}_{0}=X_0>0,\quad
dX_{s}^{*}= \frac{\lambda^2_1}{1-\gamma} X^*_s ds +\frac{\lambda_1}{1-\gamma} X^*_s dW^1_{s},\quad s\in[0,T].
\end{equation}
A standard verification argument shows that, indeed, $V$ is the value function of the optimization problem, $\pi^{*1}(X^*)$ is the optimal policy, and $X^*$ is the optimal wealth (note that $X^*$ is a geometric Brownian motion, hence, $\pi^*(X^*)\in\mathcal{A}(X_0,0)$). In particular, it follows that
$$
J(\pi)\leq J(\pi^*),
$$
for any $\pi$ that is $C$-optimal for some admissible contract $C$, with $J$ given by (\ref{eq.Princ.Obj.def}).

\subsection{Fake optimal contracts}
\label{subse:fake.opt.cont}

Recall that the notion of optimal contract used herein (cf. Definition \ref{def:OptCont}) is stronger than usual. The main additional requirement of the present definition is that the contract is robust w.r.t. capital shifts. In this subsection, we show how to construct a (trivial) contract that does not possess this feature, to illustrate the differences.

Recall the optimal wealth process of the principal, $X^*$, given by (\ref{eq.princ.opt.wealth}), and consider the following contract:
\begin{equation}\label{eq.fake.opt.cont}
\hat{C}(x):=u_0\bone_{\{X^*_T\}}(x)
\end{equation}
Note that, as long as $X^*_T$ is attainable from the current wealth level, the agent will always aim for $X^*_T$ as the terminal wealth, according to such contract. From the non-degeneracy of the market (i.e. the columns of $\sigma$ are linearly independent), it follows that the agent will keep following the prescribed strategy $\pi^*(X^*)$, given by (\ref{HJB-1}), as this is the only strategy that generates $X^*_T$. As a result, the contract $\hat{C}$ leaves both the principal and the agent satisfied. In fact, the above construction is well known in the optimal contract theory, and it always works for the first-best (risk sharing) problems. However, the resulting contract $\hat{C}$ is not robust with respect to capital injections. Indeed, if the current wealth level is perturbed, the new set of attainable terminal wealth values may not include $X^*_T$ anymore. In this case, it is not clear which strategy the agent will choose: in fact, in the case of a positive capital injection, the contract will actually provide an incentive for the agent to ``lose" (or steal) funds (which, strictly speaking, is not allowed in the model, but can certainly happen in practice). In particular, there is no guarantee that the agent will follow a strategy that is best for the principal after a capital injection is made. One can modify the definition of the ``fake" optimal contract (\ref{eq.fake.opt.cont}), by using functions other than indicator, and, e.g., obtain contracts that are non-decreasing in the terminal wealth. Nevertheless, such modifications will not resolve the main problem: the agent is not guaranteed to follow the prescribed strategy (viewed as a random field) after a capital injection is made. 

To conclude this subsection, we illustrate the importance of the fact that the individual objective of the principal, $J$, given by (\ref{eq.Princ.Obj.def}), depends on the strategy $\pi$ in a more general way than through the terminal wealth $X^{\pi}_T$ alone. Recall that the principal needs to ensure that the agent's strategy satisfies the constraint $\pi^2\equiv0$ (this is what we call an endogenous constraint). Then, if the principal's individual objective were a deterministic function of terminal wealth, e.g.,
\begin{equation*}
\tilde{J}(X^{\pi}) = \frac{1}{\gamma} \left(X^{\pi}_T\right)^{\gamma},
\end{equation*}
we could maximize the expectation of this objective, to obtain an optimal strategy $\tilde{\pi}^*$ (viewed as a random field), and choose the contract
$$
\tilde{C}(x):=\tilde{J}(x)\frac{u_0}{\EE \,\tilde{J}(X^{*}_T)}.
$$
Note that $\EE \,\tilde{J}(X^{\pi}_T)$ is indeed maximized by the desired optimal strategy $\pi^*$. The dynamic programming principle also implies that $\pi^*$ (as a random field) remains optimal for the agent, for any initial wealth level, and at any starting time. Thus, $\tilde{C}$ would be a (trivial) optimal contract, in the sense of Definition \ref{def:OptCont}. Nevertheless, this construction is only possible if the individual objective of the principal depends on $\pi$ through $X^{\pi}_T$ only. Recall, however, that, in the present formulation, $J(\pi)$ depends directly on $\pi$, via the constraint $\pi^2\equiv0$. Hence, if we use $\EE \left(X^{\pi}_T\right)^{\gamma}$ as the objective in the unconstrained problem, faced by the agent, it may not yield the same optimal strategy $\pi^*$. Indeed, the optimal contract constructed explicitly in the next subsection does not coincide with the power function with exponent $\gamma$; in fact, it becomes a random function of terminal wealth.

\subsection{Optimal contract}
\label{subse:ex.opt.cont}

Recall that $\pi^*_t(x)=(\pi^{*1}x,0)^T$, with
$$
\pi^{*1}=\frac{\lambda_1}{\sigma_1(1-\gamma)},
$$
maximizes the individual objective of the principal.
Following Proposition \ref{prop:SPDE.main} and Theorem \ref{thm:main}, we start by solving the SPDE (\ref{eq.R.SPDE.new}), which, in the present case, becomes
\begin{eqnarray*}
&&d R_t(x) = \frac{1}{2}\left[ \sigma^2_1(\pi^{*1})^2 x^2 \partial^2_{xx}R_t(x)%\right.\\
%&&\left.
+ 4\sigma^2_1(\pi^{*1})^2 x \partial_xR_t(x) + 2\sigma^2_1(\pi^{*1})^2 R_t(x)\right] dt\nonumber\\
&& - \left[ \sigma_1\pi^{*1} x \partial_x R_t(x) + \left(\lambda_1 + \sigma_1\pi^{*1}\right) R_t(x)\right] dW^1_t
- \lambda_2 R_t(x) dW^2_t,\nonumber
\end{eqnarray*}
With the ansatz $R_t(x)=R(t,x,-W^1_t,-W^2_t)$, the above becomes
$$
(\partial_t R + \frac{1}{2}\partial^2_{yy}R + \frac{1}{2}\partial^2_{zz}R) dt
- \partial_y R dW^1_t - \partial_z R dW^2_t
$$
$$
= \frac{1}{2}\left[ \sigma^2_1(\pi^{*1})^2 x^2 \partial^2_{xx}R
+ 4\sigma^2_1(\pi^{*1})^2 x \partial_xR + 2\sigma^2_1(\pi^{*1})^2 R\right] dt
$$
$$
- \left[ \sigma_1\pi^{*1} x \partial_x R + \left(\lambda_1 + \sigma_1\pi^{*1}\right) R\right] dW^1_t
- \lambda_2 R dW^2_t,
$$
which is equivalent to
$$
\partial_t R + \frac{1}{2}\partial^2_{yy}R + \frac{1}{2}\partial^2_{zz}R 
= \frac{1}{2} \sigma^2_1(\pi^{*1})^2 x^2 \partial^2_{xx}R
+ 2\sigma^2_1(\pi^{*1})^2 x \partial_xR + \sigma^2_1(\pi^{*1})^2 R,
$$
$$
\partial_y R = \sigma_1\pi^{*1} x \partial_x R + \left(\lambda_1 + \sigma_1\pi^{*1}\right) R,
\quad \partial_z R = \lambda_2 R.
$$
The following specification solves the above system:
$$
R(t,x,y,z) = \tilde{R}(t,\sigma_1\pi^{*1} y+\log x) e^{(\lambda_1 + \sigma_1\pi^{*1}) y + \lambda_2 z},
$$
$$
\partial_t \tilde{R} + A \partial_x \tilde{R} + (A+B) \tilde{R} = 0,
$$
$$
A:= \frac{1}{2}\left(2\lambda_1 \sigma_1\pi^{*1} - \sigma^2_1(\pi^{*1})^2 \right),
\quad B:= \frac{1}{2}\left( \lambda^2_1 + \lambda^2_2\right).
$$
%$$
%\partial_t \tilde{R} - f^2\partial_{xx}\tilde{R} - f \lambda_1 \partial_x \tilde{R} + \frac{1}{2} (\lambda^2_1 + \lambda^2_2) \tilde{R} = 0,
%$$
%$$
%R(t,x,y,z) = \tilde{R}(t,fy-x,z) = \tilde{R}(t,-x,z) e^{\frac{\lambda_1}{2}y } = \tilde{R}(t,-x,0) e^{\frac{\lambda_1}{2}y + \lambda_2 z}
%$$
A specific solution to the above equation is given by
$$
\tilde{R}(t,x) = \exp\left( -(B-\varepsilon A) t - (1+\varepsilon) x\right),
$$
$$
R(t,x,y,z) = \exp\left( -(B-\varepsilon A) t - (1+\varepsilon) \log x + \left( \lambda_1- \varepsilon \sigma_1\pi^{*1} \right) y + \lambda_2 z\right),
$$
with any $\varepsilon\in(0,1)$.
Then,
%$$
%R(t,x,y,z) = \exp\left( -(B-\varepsilon A) t - (1+\varepsilon) \log x - \left( \lambda_1- \varepsilon \sigma_1\pi^{*1} \right) W^1_t - \lambda_2 W^2_t\right),
%$$
$$
R_t(x) = \frac{1}{x^{1+\varepsilon}} Q_t,
$$
where
$$
Q_t = \exp\left( -(B-\varepsilon A) t - \left( \lambda_1- \varepsilon \sigma_1\pi^{*1} \right) W^1_t - \lambda_2 W^2_t\right).
$$
Let us fix any $X^*_0>0$, and note that $\lambda$, $\sigma$, $U_0$, and $\pi^*$, satisfy the assumptions of Theorem \ref{thm:main}.
To complete the construction, we choose $\bx=1$ and 
$$
a^1_t(\bx) = - \frac{\lambda_1 - \varepsilon \sigma_1 \pi^{*1}}{\varepsilon(1-\varepsilon)} Q_t,
\quad a^2_t(\bx) = - \frac{\lambda_2}{\varepsilon(1-\varepsilon)} Q_t,
$$
to obtain
\begin{equation}\label{eq.U.def.prop}
U_t(x) = \zeta_t + \int_{1}^x \int_{y}^{\infty} R_t(z) dz dy = \zeta_t + Q_t \frac{1}{\varepsilon}\int_{1}^x y^{-\varepsilon} dy
= \zeta_t + Q_t\frac{1}{\varepsilon (1-\varepsilon)} (x^{1-\varepsilon}-1),
\end{equation}
with $\zeta_0=1$ and
\begin{eqnarray*}
&& d\zeta_t = -\frac{1}{2} \sigma^2_1(\pi^{*1})^2 Q_t dt - \frac{\lambda_1 - \varepsilon \sigma_1 \pi^{*1}}{\varepsilon(1-\varepsilon)} Q_t dW^1_t - \frac{\lambda_2}{\varepsilon(1-\varepsilon)} Q_t dW^2_t
= \frac{1}{\varepsilon(1-\varepsilon)} dQ_t.
\end{eqnarray*}
Then
$$
U_t(x) = Q_t\frac{1}{\varepsilon (1-\varepsilon)} x^{1-\varepsilon},\quad C^*(x) = u_0 \left(\frac{x}{X_0}\right)^{1-\varepsilon} Q_T.
$$
Notice that such choice of $\zeta$ ensures that $U_t(x)\geq0$, for all $x>0$ and all $(t,\omega)$, thus, satisfying the limited liability condition.
In addition, we can express $Q_t$ and, hence, $U_t(x)$, as deterministic functions of the returns of the two assets, $S^1$ and $S^2$, at time $t$:
$$
W^1_t = \frac{1}{\sigma_1} \log(S^1_t/S^1_0) - \lambda_1t + \frac{\sigma_1}{2}t,
$$
$$
W^2_t = \frac{1}{\sigma_2\sqrt{1-\rho^2}}\log(S^2_t/S^2_0)
- \frac{\rho}{\sigma_1\sqrt{1-\rho^2}} \log(S^1_t/S^1_0)
+ \left( \frac{\sigma_2}{2\sqrt{1-\rho^2}} - \frac{\sigma_1\rho}{2\sqrt{1-\rho^2}} -\lambda_2\right) t,
$$
%$$
%- (\lambda_1 - \varepsilon \sigma_1 \pi^{*1}) W^1_t - \lambda_2 W^2_t
%= \left(\frac{\rho\lambda_2}{\sigma_1\sqrt{1-\rho^2}} - \frac{\lambda_1 - \varepsilon \sigma_1 \pi^{*1}}{\sigma_1}\right) \log(S^1_t/S^1_0)
%$$
%$$
%- \frac{\lambda_2}{\sigma_2\sqrt{1-\rho^2}}\log(S^2_t/S^2_0)
%+ \left( \lambda_1^2 + \lambda_2^2 - \varepsilon \lambda_1 \sigma_1 \pi^{*1} - \lambda_1 \frac{\sigma_1}{2}
%+ \varepsilon \pi^{*1} \frac{\sigma^2_1}{2} 
%- \lambda_2 \frac{\sigma_2 - \sigma_1\rho}{2\sqrt{1-\rho^2}} \right) t,
%$$
$$
Q_t = \exp\left( \left( \frac{1}{2}(\lambda_1^2 + \lambda_2^2) - \lambda_1 \frac{\sigma_1}{2}
+ \varepsilon \pi^{*1} \frac{\sigma^2_1}{2} (1-\pi^{*1})
- \lambda_2 \frac{\sigma_2 - \sigma_1\rho}{2\sqrt{1-\rho^2}} \right) t \right)
$$
$$
\times \left(\frac{S^1_t}{S^1_0}\right)^{\varepsilon\pi^{*1} + \frac{\rho\lambda_2}{\sigma_1\sqrt{1-\rho^2}} - \frac{\lambda_1}{\sigma_1}}
\left(\frac{S^2_t}{S^2_0}\right)^{- \frac{\lambda_2}{\sigma_2\sqrt{1-\rho^2}}}
:=\widehat{Q}\left(t,S^2_t/S^2_0,S^3_t/S^3_0\right).
$$
To conclude that $C^*$ is an optimal contract, it remains to verify that the assumptions of the last statement of Theorem \ref{thm:main} are satisfied. Note that $U\geq0$. Part 1 of Theorem \ref{thm:main} implies that, for any $(\xi,\tau)\in\mX$ and any $\pi\in\mathcal{A}(\xi,\tau)$, the process $\left(U_t\left(X^{\pi,\xi,\tau}_t\right)\right)_{t\in[\tau,T]}$ is a local supermartingale. As it is nonnegative, and 
$$
U_{\tau}(\xi)=\text{const}\cdot Q_{\tau}\, \xi^{1-\varepsilon}\in L^1,
$$
(which follows form H\"older inequality), an application of Fatou's lemma yields that it is a true supermartingale.
%Hence, it can be decomposed into a sum of a non-increasing process $A$, with $A_0=U_0(\xi)$ and a local martingale $M$, with $M_0=0$. The nonnegativity implies
%$$
%M_t \geq - A_t \geq -U_0(\xi),
%$$
%and, as $U_0(\xi)=\text{const} \xi^{1-\varepsilon}$ and $\xi\in L^1$. Hence, the right hand side of the above is absolutely integrable. Applying Fatou's lemma again, we conclude that $M$ is a supermartingale.
Next, Part 2 of Theorem \ref{thm:main} implies that $\left(U_t\left(X^{*,\xi,\tau}_t\right)\right)_{t\in[\tau,T]}$ is a local martingale. As it is also positive, we have
$$
\EE \sup_{t\in[0,T]}\left| U_t\left(X^{*,\xi,\tau}_t\right) \right| \leq \text{const}\cdot \EE \left(\xi^{1-\varepsilon} \sup_{t\in[0,T]} \left(Q_t \left(X^{*,1,\tau}_t\right)^{1-\varepsilon}\right) \right)<\infty,
$$
which follows, again, from H\"older inequality, by observing that the expression inside the supremum is a geometric Brownian motion.
The above inequality implies that $\left(U_t\left(X^{*,\xi,\tau}_t\right)\right)_{t\in[\tau,T]}$ is a true martingale and completes the proof of the fact that $C^*$ is an optimal contract (by Theorem \ref{thm:main}).

Notice that the optimal contract $C^*$ is given by a power function of terminal wealth multiplied by a random scalar. This is in contrast to the individual objective of the principal, which is a deterministic function of terminal wealth. The random scalar, $Q_T$, itself, is a power function of the returns generated by the two assets available in the market. Thus, effectively, the optimal contract measures the terminal wealth generated by the agent relative to the performance of the available assets. Note also that the exponents in the latter power functions depend on the characteristics of the assets, such as the market price of risk.
%In addition, these exponents are computed differently for the two assets (i.e., even if the two assets have the same characteristics and zero correlation, the exponents would be different), which reflects the fact that the target strategy of the principal is not symmetric: it prescribes investing in the first asset only.
Recall also that the optimal contract is nonnegative, thus, satisfying the limited liability condition. 

Note also that, as $\varepsilon\approx 0$, the optimal contract converges to $u_0$ multiplied by
$$
\frac{x/X_0}{\widehat{Q}\left(T,S^2_T/S^2_0,S^3_T/S^3_0\right)}.
$$ 
The above ratio measures the return of the fund relative to the returns of the two assets, the latter being captured by $\widehat{Q}$. If this ratio exceeds one (i.e. if the fund outperformance the benchmark), the manager's fee exceeds its initially expected value $u_0$ (i.e. he receives a bonus). Otherwise, his payment drops below $u_0$ (i.e. he is penalized).

Finally, it is worth mentioning that the optimal contract $C^*$ is a deterministic function of the terminal values of the wealth process and of the tradable assets. Hence, it also provides a solution to the second-best (moral hazard) version of the problem, in which the principal only observes $(X^{\pi}_T,S^1_T,S^2_T)$ and, hence, cannot fully deduce the agent's strategy $\pi$ from her observations. This is not surprising, however, since the terminal value of the target optimal wealth process, $X^*_T$, is a deterministic function of $(S^1_T,S^2_T)$, which means that observing the latter values is sufficient for the principal to enforce the desired trading strategy.

\bibliographystyle{plain}
\bibliography{OptContractForward_refs}

\end{document}